%% file: main.tex
\documentclass[11pt, a4paper]{article}
\usepackage[letterpaper, left=1in, right=1in, top=0.9in, bottom=0.9in]{geometry}

\usepackage{amsmath}
\usepackage{amssymb}
\usepackage{amsfonts}
\usepackage{amsthm} 
\usepackage{graphicx}
\usepackage{comment}

\usepackage{thm-restate}

\usepackage{hyperref}
\usepackage{here}

\usepackage{cite}

\newtheorem{theorem}{Theorem}[section]
\newtheorem{lemma}[theorem]{Lemma}
\newtheorem{proposition}[theorem]{Proposition}

\newtheorem{definition}[theorem]{Definition}
\newtheorem{remark}[theorem]{Remark}

\usepackage[linesnumbered,noend,ruled,vlined]{algorithm2e}

\SetKwInput{KwInput}{Input}     
\SetKwInput{KwOutput}{Output}   
\SetKw{KwTo}{to}

\newcommand{\algorithmicbreak}{\textbf{break}}
\newcommand{\Break}{\algorithmicbreak}

\global\long\def\E{\mathbb{E}}%
\global\long\def\M{\mathcal{M}}%
\global\long\def\I{\mathcal{I}}%
\global\long\def\B{\mathcal{B}}%
\global\long\def\tO{\tilde{O}}%
\global\long\def\TO{\tilde{O}_{\varepsilon}}%

\usepackage{tabularx}
\usepackage{environ}

\usepackage{cite}
\usepackage{url}

\usepackage[normalem]{ulem}
\usepackage{color}

\title{Subquadratic Submodular Maximization \\ with a General Matroid Constraint\thanks{
The authors thank the three anonymous reviewers for their valuable comments.
This work was partially supported by the joint project of Kyoto University and Toyota Motor Corporation, titled ``Advanced Mathematical Science for Mobility Society'', by JST ERATO Grant Number JPMJER2310, and by JSPS
KAKENHI Grant Numbers JP20K11692, JP22H05001, JP24KJ1494, and JP24K02901.
}}
\author{Yusuke Kobayashi\thanks{
Research Institute for Mathematical Sciences, Kyoto University.
E-mail: \{yusuke, ttatsuya\}@kurims.kyoto-u.ac.jp
}
\and Tatsuya Terao\footnotemark[2]
}

\date{}

\begin{document}

\maketitle

\begin{abstract}
We consider fast algorithms for monotone submodular maximization with a general matroid constraint. We present a randomized $(1 - 1/e - \epsilon)$-approximation algorithm that requires $\tilde{O}_{\epsilon}(\sqrt{r} n)$ independence oracle and value oracle queries, where $n$ is the number of elements in the matroid and $r \leq n$ is the rank of the matroid. This improves upon the previously best algorithm by Buchbinder-Feldman-Schwartz [Mathematics of Operations Research 2017] that requires $\tilde{O}_{\epsilon}(r^2 + \sqrt{r}n)$ queries.

Our algorithm is based on continuous relaxation, as with other submodular maximization algorithms in the literature. To achieve subquadratic query complexity, we develop a new rounding algorithm, which is our main technical contribution. The rounding algorithm takes as input a point represented as a convex combination of $t$ bases of a matroid and rounds it to an integral solution. Our rounding algorithm requires $\tilde{O}(r^{3/2} t)$ independence oracle queries, while the previously best rounding algorithm by Chekuri-Vondr\'{a}k-Zenklusen [FOCS 2010] requires $O(r^2 t)$ independence oracle queries. A key idea in our rounding algorithm is to use a directed cycle of arbitrary length in an auxiliary graph, while the algorithm of Chekuri-Vondr\'{a}k-Zenklusen focused on directed cycles of length two. 
\end{abstract}

\input{intro.tex}

\input{preliminaries.tex}

\input{submodular_maximization.tex}

\input{algorithm.tex}

\input{appendix_other_oracle.tex}

\input{conclusion.tex}


\bibliography{biblio}

\end{document}

%% file: intro.tex
\section{Introduction} \label{sec:intro}

\subsection{Submodular Maximization}
\label{subsec:submodularmax}

Submodular maximization is a fundamental and well-studied problem in theoretical computer science and combinatorial optimization.
This is because a number of important problems can be regarded as special cases of submodular maximization, including maximum coverage, generalized assignment, and facility location.
Furthermore, submodular maximization has many practical applications in machine learning, economics, and many other areas.
In the submodular maximization problem, the input consists of 
a (monotone) submodular set function $f \colon 2^V \to \mathbb{R}_+$ 
and a feasible region ${\cal F} \subseteq 2^V$ specified by some constraints, 
and the aim is to find a set $S \in {\cal F}$ maximizing $f(S)$. 

The study of submodular maximization was initiated by a seminal work of Fisher, Nemhauser, and Wolsey in the 1970's \cite{nemhauser1978analysis, fisher1978analysis, nemhauser1978best}.
They showed that, for monotone submodular maximization, the greedy algorithm achieves $(1 - 1/e)$-approximation for a cardinality constraint and $\frac{1}{2}$-approximation for a matroid constraint.
The advantage of their algorithm is that it is very simple and fast, indeed, it runs in quadratic time.
It is known that unless $\rm{P} = \rm{NP}$, for any $\varepsilon > 0$, there is no $(1 - 1/e + \varepsilon)$-approximation algorithm for the maximum coverage problem \cite{feige1998threshold}, which is a special case of monotone submodular maximization with a cardinality constraint or a matroid constraint.
Thus, the factor $1 - 1/e$ is optimal for a cardinality constraint.

To obtain an optimal $(1 - 1/e)$-approximation algorithm for a matroid constraint, Calinescu-Chekuri-P\'{a}l-Vondr\'{a}k~\cite{calinescu2011maximizing} developed a framework based on continuous optimization and rounding technique.
In their algorithm, they first solve the continuous optimization problem of maximizing the multilinear extension $F$ of $f$, a natural continuous extension of $f$.
By using a continuous greedy algorithm, they obtain a $(1 - 1/e)$-approximation solution for the continuous optimization problem.
In order to round the obtained fractional solution to an integral one, they use a variant of the pipage rounding technique of Ageev-Sviridenko~\cite{ageev2004pipage}. 
Consequently, their algorithm achieves the optimal $(1 - 1/e)$-approximation.
%
%
Note that, although their algorithm runs in polynomial time, its running time is very high. 

Since submodular maximization has a number of applications, 
providing efficient approximation algorithms is a fundamental task both in theory and in practice. 
Thus, it has received considerable attention to develop fast
submodular maximization algorithms 
that achieve an approximation close to the optimal factor, 
typically with an approximation factor of $1 - 1/e - \varepsilon$ for any $\varepsilon > 0$.

In the submodular maximization problem with a general matroid constraint, 
it is standard to suppose that the objective function $f$ is given as a value oracle, 
and the feasible region ${\cal F} \subseteq 2^V$ is given as an independence oracle of a matroid. 
In such a case, the efficiency of an algorithm is 
usually measured by the number of value and independence oracle queries used in it. 

Badanidiyuru-Vondr\'{a}k~\cite{badanidiyuru2014fast} presented a fast algorithm that achieves an almost optimal approximation factor $1 - 1/e - \varepsilon$, for any $\varepsilon > 0$, for a matroid constraint.
Their algorithm uses $O\left( \frac{rn}{\varepsilon^4} \log^2 \left( \frac{n}{\varepsilon} \right) \right)$ value oracle queries and $O\left( \frac{n}{\varepsilon^2} \log \left( \frac{n}{\varepsilon} \right) + \frac{r^2}{\varepsilon} \right)$ independence oracle queries, where $n$ is the number of elements in the matroid and $r$ is the rank of the matroid.
To achieve this query complexity, 
they developed a fast implementation of the continuous greedy algorithm that uses $\TO(rn)$ value oracle queries and $\TO(n)$ independence oracle queries.\footnote{The $\tilde{O}_{\varepsilon}$ notation hides polylogarithmic factors in $n$ and polynomial factors in $\varepsilon^{-1}$.}
The output of their continuous greedy algorithm is
a fractional solution represented as a convex combination of $1 / \varepsilon$ bases.
Then, they apply the swap rounding algorithm of Chekuri-Vondr\'{a}k-Zenklusen~\cite{chekuri2010dependent} to round the obtained fractional solution to an integral solution, 
which requires $O(r^2 / \varepsilon)$ independence oracle queries. 



Buchbinder-Feldman-Schwartz~\cite{buchbinder2017comparing} presented a $(1 - 1/e - \varepsilon)$-approximation algorithm that has a trade-off between the number of value oracle queries and the number of independence oracle queries used in the algorithm.
In their algorithm, they combine a variant of the residual random greedy algorithm of Buchbinder-Feldman-Naor-Schwartz~\cite{buchbinder2014submodular}  and the fast continuous greedy algorithm of Badanidiyuru-Vondr\'{a}k described above.
Then, for a parameter $\lambda \in [1, r]$, their algorithm uses $\TO(r \lambda + \frac{rn}{\lambda})$ value oracle queries and $\TO(\lambda n + r^2)$ independence oracle queries.
We note that the $\TO(r^2)$ term in the independence query complexity is due to the rounding algorithm in the same way as the algorithm of Badanidiyuru-Vondr\'{a}k.
If we evaluate the algorithm by the total number of queries regardless of their types, 
then the query complexity is minimized when $\lambda = \Theta(\sqrt{r})$. 
In this case, their algorithm uses $\tilde{O}_{\varepsilon}(r^2 + \sqrt{r}n)$ value and independence oracle queries.
This query complexity is better than that of Badanidiyuru-Vondr\'{a}k~\cite{badanidiyuru2014fast} when $r = o(n)$, 
but a quadratic number of queries is still required when $r$ is large. 





Recently, for several important classes of matroids, 
faster algorithms for monotone submodular maximization with a matroid constraint have been investigated. 
Ene-Nguy$\tilde{{\hat{\text{e}}}}$n~\cite{ene2018towards} presented a $(1 - 1/e - \varepsilon)$-approximation algorithm for graphic matroid and partition matroid constraints in time nearly-linear in the size of their representation.
Henzinger-Liu-Vondr\'{a}k-Zheng~\cite{henzinger2023faster} presented a $(1 - 1/e - \varepsilon)$-approximation algorithm for laminar matroid and transversal matroid constraints in nearly-linear time.
A key ingredient in these algorithms is a fast dynamic data structure for maintaining an (approximate) maximum weight basis of the matroid.

\subsection{Our Results}
\label{subsec:our_contribution}

This paper focuses on the monotone submodular maximization problem 
with a general matroid constraint. 
In the problem, this input consists of 
a monotone submodular set function $f \colon 2^V \to \mathbb{R}_+$ given as a value oracle, and 
a matroid $\M = (V, \I)$ given as an independence oracle. 
The objective is to find an independent set $S \in \I$ that maximizes $f(S)$.  
For $\alpha \in [0, 1]$, a randomized algorithm is said to be an {\em $\alpha$-approximation algorithm} 
if it returns a solution $S \in \I$ with $\E[f(S)] \ge \alpha \cdot \max \{ f(T) \mid T \in \I \}$.  
A randomized algorithm is often called simply an algorithm throughout the paper.  
Our main result is to give a first $(1 - 1/e - \varepsilon)$-approximation algorithm for this problem 
that requires a subquadratic number of queries. 
Recall that $n = |V|$ and $r$ is the rank of $\M$.

\begin{theorem} \label{main_submodular_max}
For any $\varepsilon > 0$, there is a randomized algorithm that achieves $(1 - 1/e - \varepsilon)$-approximation for maximizing a monotone submodular function subject to a matroid constraint and uses $O(\sqrt{r} n \text{ \textup{poly}}(1/\varepsilon, \log n))$ value and independence oracle queries.
\end{theorem}

It is worth mentioning that, for the case of $r = \Theta(n)$, our algorithm uses $\TO(n^{3/2})$ oracle queries, whereas the algorithm of Buchbinder-Feldman-Schwartz~\cite{buchbinder2017comparing} uses $\TO(n^2)$ oracle queries.

Our algorithm is based on continuous relaxation
and rounding technique in the same way as previous algorithms~\cite{badanidiyuru2014fast,buchbinder2017comparing,calinescu2011maximizing}. 
In this framework, currently, the bottleneck of the query complexity comes from the rounding algorithm. 
Indeed, the swap rounding algorithm by 
Chekuri-Vondr\'{a}k-Zenklusen~\cite{chekuri2010dependent} 
requires $O(r^2 t)$ independence oracle queries 
if the input point is represented as a convex combination of $t$ bases of the matroid. 
Then, this rounding algorithm requires a quadratic number of independence oracle queries even when 
$t$ is small. 
Therefore, 
in order to break the quadratic-independence-query barrier in this framework, 
it is necessary to devise a faster rounding algorithm. 

The key technical contribution of this paper is to 
develop a new rounding algorithm that uses $o(r^2t)$ independence oracle queries. 

\begin{restatable}{theorem}{rounding}
\label{thm:main_rounding}
For any $\varepsilon > 0$, there is a randomized algorithm satisfying the following conditions: 
\begin{itemize}
    \item the input consists of a matroid $\M = (V, \I)$ given as an independence oracle and 
    a point $x$ in the base polytope of $\M$ 
    that is represented as a convex combination of $t$ bases, 
    \item the output is a basis $S$ of $\M$ such that $\E[f(S)] \geq (1 - \varepsilon) F(x)$ for any submodular function $f\colon 2^V \to \mathbb{R}$ and its multilinear extension $F$, and 
    \item it uses $O(r^{3/2}t \log^{3/2}(\frac{r t}{\varepsilon}) )$ independence oracle queries.  
\end{itemize}  
\end{restatable}


By combining this theorem with 
the submodular maximization algorithm by Buchbinder-Feldman-Schwartz~\cite{buchbinder2017comparing}, 
we obtain Theorem~\ref{main_submodular_max}; see Section~\ref{sec:submodular_maximization} for details. 

We also show that if the matroid is given as a rank oracle instead of an independence oracle, then 
we obtain a $(1 - 1/e - \varepsilon)$-approximation algorithm using $\TO(n + r^{3/2})$ value and rank oracle queries. 

\begin{restatable}{theorem}{ranksubmodularmax} \label{rank_submodular_max}
For any $\varepsilon > 0$, there is a randomized algorithm that achieves $(1 - 1/e - \varepsilon)$--approximation for maximizing a monotone submodular function subject to a matroid constraint and uses $O((n + r^{3/2}) \text{ \textup{poly}}(1/\varepsilon, \log n))$ value and rank oracle queries.
\end{restatable}

\subsection{Overview of Our Rounding Algorithm}
\label{subsec:overview}










In this subsection, 
we give a technical overview of our new rounding algorithm. 
Since our rounding algorithm is based on that of Chekuri-Vondr\'{a}k-Zenklusen~\cite{chekuri2010dependent}, 
we first review their algorithm
and then explain the key ideas behind ours.

\paragraph*{Swap Rounding Algorithm of Chekuri-Vondr\'{a}k-Zenklusen.}
The rounding algorithm by Chekuri-Vondr\'{a}k-Zenklusen is called the swap rounding algorithm. 
Their algorithm takes as input a point $x$ represented as a convex combination of $t$ bases of $\M$ and returns an integral solution $S$ such that $\E [f(S)] \geq F(x)$ for any submodular function $f$ and its multilinear extension $F$.
In each phase of the algorithm, we pick up two bases in the representation of $x$ and merge them into a single basis.
By applying this procedure $t - 1$ times, we obtain a single basis of $\M$.

In order to merge two bases, say $B_1$ and $B_2$, 
their swap rounding algorithm uses a {\em strongly exchangeable pair} of elements, that is, 
a pair of elements $u \in B_1 \setminus B_2$ and $v \in B_2 \setminus B_1$ such that $B_1 + v - u \in \I$ and $B_2 + u - v \in \I$.
Since we can find a strongly exchangeable pair using $O(r)$ independence oracle queries 
and we need to find such a pair $O(rt)$ times in the algorithm, 
the total number of queries is $O(r^2 t)$. 
%
It is still not clear whether we can develop an algorithm for finding a strongly exchangeable pair using $o(r)$ independence oracle queries, 
and hence their algorithm is now stuck at $\Omega(r^2 t)$ independence oracle queries.

See Section~\ref{sec:preimplement} for details of the swap rounding algorithm of Chekuri-Vondr\'{a}k-Zenklusen.

\paragraph*{Our Faster Rounding Algorithm.}

We develop a new rounding algorithm that requires $\tilde{O}(r^{3/2} t)$ independence oracle queries with high probability.
Our rounding algorithm is based on that of Chekuri-Vondr\'{a}k-Zenklusen 
in a sense that we update bases by swapping a pair of elements $O(rt)$ times. 
Therefore, in each step of our algorithm, we need to update some basis 
by using only $\tilde{O}(\sqrt{r})$ independence oracle queries. 
To achieve this, we need substantially new ideas. 


First, we introduce a digraph that represents exchangeability of the elements in the matroid (see Definition \ref{def:auxiliary_directed_graph}), 
and provide a new interpretation of the swap rounding algorithm of Chekuri-Vondr\'{a}k-Zenklusen using this auxiliary graph.  
Indeed, each step of their algorithm can be seen as an update 
using a directed cycle of length two in the auxiliary graph. 
This interpretation motivates us to focus on a directed cycle of arbitrary length in the auxiliary graph
instead of a directed cycle of length two.  
By extending the argument of Chekuri-Vondr\'{a}k-Zenklusen, 
we show that we can appropriately update bases using a directed cycle of arbitrary length in the auxiliary graph.


Second, we show that we can find a directed cycle in the auxiliary graph using $o(r)$ independence oracle queries with high probability, 
which is the most technical part in our argument. 
To achieve this, we combine sampling technique and binary search technique. 
In our algorithm for finding a directed cycle in the auxiliary graph, 
we first sample $\tilde{O}(\sqrt{r})$ vertices, and 
define $D'$ as the subgraph induced by the sampled vertex set. 
If every vertex in $D'$ has an incoming edge, then we can easily find a directed cycle in $D'$ by traversing such directed edges in the opposite direction. 
Otherwise, by using a vertex with no incoming edge, 
we find a directed cycle of length two using $\tilde{O}(\sqrt{r})$ independence oracle queries with high probability.
We can check whether each vertex in $D'$ has an incoming edge or not using $\tilde{O}(1)$ independence oracle queries
with the aid of the binary search technique 
proposed by Nguy{\~{\^{e}}}n~\cite{nguyen2019note} and Chakrabarty-Lee-Sidford-Singla-Wong~\cite{chakrabarty2019faster}; see Lemma \ref{lem:binary_search_ind} for details. 
Note that this technique was 
used in recent studies on fast matroid intersection \cite{nguyen2019note, chakrabarty2019faster, blikstad2021breaking, blikstad2021breaking_STOC, u2022subquadratic} and matroid partition \cite{terao2023faster} algorithms.
Therefore, we obtain an algorithm that finds a directed cycle using $\tO(\sqrt{r})$ independence oracle queries with high probability.

In our rounding algorithm, we update bases using a directed cycle in the auxiliary graph repeatedly. 
Since we update bases $O(rt)$ times and each update requires $\tO(\sqrt{r})$ independence oracle queries, 
the total number of independence oracle queries is $\tilde{O}(r^{3/2} t)$ with high probability.

\subsection{Related Work} \label{subsec:related_work}

We have mentioned several recent studies on fast submodular maximization 
with matroid constraints in Section~\ref{subsec:submodularmax}. 
Other than these, there are a lot of studies on fast submodular maximization algorithms in the literature \cite{azar2012efficient, badanidiyuru2014fast, chekuri2015multiplicative, li2022submodular, filmus2012tight, ene2017nearly, mirzasoleiman2015lazier}. 

Badanidiyuru-Vondr\'{a}k~\cite{badanidiyuru2014fast} developed a $(1 - 1/e - \varepsilon)$-approximation algorithm using $O(\frac{n}{\varepsilon} \log (\frac{n}{\varepsilon}))$ value oracle queries for the cardinality constraint.
Mirzasoleiman-Badanidiyuru-Ashwinkumar-Karbasi-Vondr{\'a}k-Krause~\cite{mirzasoleiman2015lazier} developed a $(1 - 1/e - \varepsilon)$-approximation algorithm using $O(n \log (1 / \varepsilon))$ value oracle queries for the cardinality constraint.
Ene-Nguy$\tilde{{\hat{\text{e}}}}$n~\cite{ene2017nearly} developed a $(1 - 1/e - \varepsilon)$-approximation algorithm using $(1/\varepsilon)^{O(1/\varepsilon^4)} n \log^2 n$ value oracle queries for the knapsack constraint.
Filmus-Ward~\cite{filmus2012tight} presented a combinatorial $(1 - 1/e)$-approximation algorithm for monotone submodular maximization with a matroid constraint, which uses $O(n^7r^2)$ oracle queries.
They also obtain a $(1 - 1/e - O(\varepsilon))$-approximation algorithm that uses $O(\varepsilon^{-3} n^4 r)$ value oracle queries and $O(\varepsilon^{-1}n^2r \log n)$ independence oracle queries.


Studies on fast submodular maximization algorithms have developed also in the direction of parallelized settings \cite{balkanski2018adaptive, balkanski2019optimal, ene2019submodular, ene2019submodular_STOC2019, chekuri2019parallelizing}, distributed settings \cite{barbosa2016new, liu_et_al:OASIcs.SOSA.2019.18}, and dynamic settings \cite{chen2022complexity, banihashem2024dynamic, lattanzi2020fully, monemizadeh2020dynamic}.

Chekuri-Quanrud-Torres~\cite{chekuri2021fast} developed a fast swap rounding algorithm for graphic matroid constraints to obtain fast approximation algorithms for the Bounded Degree MST problem and the Crossing Spanning Tree problem.

\subsection{Paper Organization} \label{subsec:organization}

The remaining of this paper is organized as follows. 
In Section~\ref{sec:preliminaries}, we give some preliminaries. 
In Section~\ref{sec:submodular_maximization}, 
we show how to derive Theorem \ref{main_submodular_max} from 
our fast rounding algorithm in Theorem~\ref{thm:main_rounding}. 
In Section \ref{sec:preimplement}, 
we describe the swap rounding algorithm by Chekuri-Vondr\'{a}k-Zenklusen~\cite{chekuri2010dependent} in detail,
because it is the basis of our rounding algorithm. 
In Section~\ref{sec:faster_rounding_algorithm}, 
we describe our fast rounding algorithm and prove Theorem~\ref{thm:main_rounding}, 
which is the main technical part of this paper. 
In Section~\ref{sec:other_oracle}, 
we discuss the rank oracle setting and prove 
Theorem \ref{rank_submodular_max}.

%% file: preliminaries.tex
\section{Preliminaries} \label{sec:preliminaries}


\paragraph*{Basic Notation.}
Let $\mathbb{R}_{+}$ denote the set of non-negative real numbers. 
Throughout this paper, let $V$ be a finite set and let $n$ denote its cardinality. 
For a set $A \subseteq V$ and an element $v \in V$, 
we will often write $A + v := A \cup \{ v \}$ and $A - v := A \setminus \{ v \}$.
For two sets $A, B \subseteq V$, their symmetric difference 
is denoted by $A \triangle B := (A \setminus B) \cup (B \setminus A)$. 
For $A \subseteq V$, the {\em characteristic vector} of $A$ is defined as the vector $x \in \{0, 1\}^V$ with $x_v = 1$ for $v \in A$ and $x_v = 0$ for $v \in V \setminus A$. 
We will denote by $\mathbf{1}_{A}$ the characteristic vector of $A$.
For $v \in V$, we will write $\mathbf{1}_{v} := \mathbf{1}_{\{ v \}}$.

\paragraph*{Submodular Functions and Multilinear Extension.}
Let $f \colon 2^V \rightarrow \mathbb{R}_{+}$ be a set function on a finite ground set $V$ of size $n$.
The function is \emph{submodular} if $f(A) + f(B) \geq f(A \cup B) + f(A \cap B)$ for any two subsets $A, B \subseteq V$. 
The function is \emph{monotone} if $f(A) \leq f(B)$ for any subsets $A \subseteq B \subseteq V$.
In this paper, we only consider monotone submodular functions.

For a function $f \colon 2^V \rightarrow \mathbb{R}_{+}$,
we define its {\em multilinear extension} $F \colon [0, 1]^V \rightarrow \mathbb{R}_{+}$ by
\begin{equation*}
F(x) = \sum_{S \subseteq V} f(S) \prod_{v \in S} x_v \prod_{v \in V \setminus S} (1 - x_v)
\end{equation*}
for $x \in [0, 1]^V$. 
Note that this value is equal to $\E[f(R(x))]$,
where $R(x)$ is a random set that contains each element $v \in V$ independently with probability $x_v$.
In particular, $F(\mathbf{1}_{S}) = f(S)$ for any $S\subseteq V$.

\paragraph*{Matroids.}
A pair $\M = (V, \mathcal{I})$ of a finite set $V$ and a non-empty set family $\mathcal{I} \subseteq 2^{V}$ is called a {\em matroid} if the following properties are satisfied.

\begin{description}
\item[(Downward closure property)] 
if $S \in \mathcal{I}$ and $S' \subseteq S$, then $S' \in \mathcal{I}$.

\item[(Augmentation property)]
if $S, S' \in \mathcal{I}$ and $|S'| < |S|$, then there exists $v \in S \setminus S'$ such that $S' + v \in \mathcal{I}$.
\end{description}

A set $S \subseteq V$ is called {\em independent} if $S \in \mathcal{I}$ and {\em dependent} otherwise.
The {\em rank} of $\M$ is defined as the size of a largest independent set.
In addition, for a subset $S \subseteq V$, the rank of $S$ is defined as the size of a largest independent set contained in $S$.
Inclusionwise maximal independent sets are called {\em bases}.
Note that every basis has the same size. 
For an independent set $S \in \I$, let $\M/S = (V \setminus S, \I')$ be the matroid obtained by contracting $S$ in $\M$, that is, $S' \in \I'$ if and only if $S' \cup S \in \I$.

Let $\B$ be the set of all bases of a matroid $\M = (V, \mathcal{I})$ and 
let $B, B' \in \B$ be two bases.
It is well-known that, for any $u \in B \setminus B'$, there exists $v \in B' \setminus B$ such that $B - u + v \in \B$ and $B' - v + u \in \B$ (see e.g.,~\cite[Theorem 39.12]{schrijver2003combinatorial}). 
This property is called {\em strong basis exchange property} of matroids. 


Let $\M = (V, \I)$ be a matroid whose rank function and basis family 
are denoted by $r_{\M}$ and $\B$, respectively. 
The {\em matroid polytope} $P(\M)$  
is defined as the convex hull of the characteristic vectors of all the independent sets of $\M$. The {\em matroid base polytope} $B(\M)$ is defined as the convex hull of the characteristic vectors of all the bases of $\M$.
It is well-known that $P(\M)$ and $B(\M)$ are described as follows (see e.g.,~\cite[Section 40.2]{schrijver2003combinatorial}): 
\begin{align*}
P(\M) &:= \text{conv}\{ \mathbf{1}_{I} \mid I \in \I \} = \left\{x \in \mathbb{R}_+^V \; \middle| \; \sum_{v \in S} x_v \leq r_{\M}(S) \text{ for any } S \subseteq V  \right\}, \\
B(\M) &:= \text{conv}\{ \mathbf{1}_{B} \mid B \in \B \} = \left\{x \in P(\M) \; \middle| \; \sum_{v \in V} x_v = r_{\M}(V) \right\}.
\end{align*}

\paragraph*{Oracles.}
When we consider the submodular maximization problem, 
we assume that the submodular function $f$ is given as a value oracle, which takes as input any subset $S \subseteq V$ and outputs $f(S)$.
We also assume that we access a matroid $\M$ through an oracle.
Given a subset $S \subseteq V$, an {\em independence oracle} outputs whether $S \in \I$ or not.
Given a subset $S \subseteq V$, a {\em rank oracle} outputs the rank of $S$, i.e., the size of a largest independent set contained in $S$.
Note that the rank oracle is more powerful than the independence oracle, since one query of the rank oracle can determine whether a given subset is independent or not.

\paragraph*{Binary Search Technique.}

For a matroid $\M = (V, \I)$, an independent set $S \in \I$, an element $u \in V \setminus S$, and $T \subseteq S$, 
we consider a procedure that finds
an element $v \in T$ with $S + u - v \in \I$ if one exists. 
Chakrabarty et al.~\cite{chakrabarty2019faster} and Nguy{\~{\^{e}}}n~\cite{nguyen2019note} independently 
proved that this procedure can be implemented efficiently 
by using the binary search technique in the independence oracle model.
Their result is formally described as follows. 


\begin{lemma}[\cite{nguyen2019note,chakrabarty2019faster}]\label{lem:binary_search_ind}
There is an algorithm \textup{\texttt{FindExchangeElement}} which, 
given a matroid $\M = (V, \mathcal{I})$, an independent set $S \in \mathcal{I}$, an element $u \in V \setminus S$, and $T \subseteq S$, finds an element $v \in T$ such that $S + u - v \in \mathcal{I}$ or otherwise determines that no such element exists, and uses $O(\log |T|)$ independence oracle queries.
\end{lemma}

%% file: submodular_maximization.tex
\section{Submodular Maxmization Algorithm (Proof of Theorem \ref{main_submodular_max})} \label{sec:submodular_maximization}

In this section, we give a proof of Theorem \ref{main_submodular_max} by combining 
the algorithm of Buchbinder-Feldman-Schwartz~\cite{buchbinder2017comparing}
and 
our rounding algorithm in Theorem~\ref{thm:main_rounding}. 
Note that a proof of Theorem~\ref{thm:main_rounding} is given in Section~\ref{sec:faster_rounding_algorithm} later. 



For monotone submodular maximization with a matroid constraint,
Buchbinder-Feldman-Schwartz presented 
a $(1 - 1/e - \varepsilon)$-approximation algorithm that has a trade-off between the number of value oracle queries 
and the number of independence oracle queries used in the algorithm. 
The main part of their algorithm is to solve 
the continuous relaxation of the submodular maximization problem efficiently. 

Let $\lambda \in [1, r]$ be a parameter that controls the trade-off. 
In their algorithm for solving the continuous relaxation problem, 
they first apply a variant of the residual random greedy algorithm of Buchbinder-Feldman-Naor-Schwartz~\cite{buchbinder2014submodular}.
This residual random greedy algorithm outputs $S \subseteq V$ and uses $\TO(r \lambda  + n)$ value oracle queries and $\TO(\lambda n)$ independence oracle queries; 
see \cite[Lemma 3.3]{buchbinder2017comparing}.
Then they apply a variant of the fast continuous greedy algorithm of Badanidiyuru-Vondr\'{a}k~\cite{badanidiyuru2014fast}.
This continuous greedy algorithm outputs a point $x'$ represented as a convex combination of $1/\varepsilon$ bases of $\M/S$ and uses $\TO(\frac{rn}{\lambda})$ value oracle queries and $\TO(n)$ independence oracle queries; see \cite[Corollary 3.1]{buchbinder2017comparing}.
Then $x = \mathbf{1}_S \vee x'$ is an approximate solution for the continuous relaxation problem, 
which can be represented as a convex combination of $1/\varepsilon$ bases of $\M$.   
Here, for vectors $y$ and $z$, let $y \vee z$ denote the vector such that $(y \vee z)_i = \max \{ y_i, z_i \}$ for all $i$.

Overall, Buchbinder-Feldman-Schwartz~\cite{buchbinder2017comparing}
presented an efficient algorithm for solving the continuous relaxation problem, 
which is formally stated as follows.

\begin{theorem}[{follows from \cite[Corollary 3.1]{buchbinder2017comparing} and \cite[Lemma 3.3]{buchbinder2017comparing}}] \label{submodular_max_lemma_independence}
Given a non-negative monotone submodular function $f \colon 2^{V} \rightarrow \mathbb{R}_+$, a matroid $\M = (V, \I)$ of rank $r$, and parameters $\varepsilon > 0$ and $\lambda \in [1, r]$,
there is an algorithm satisfying the following conditions: 
\begin{itemize}
\item the algorithm outputs a point $x \in B(\M)$ represented as a convex combination of $1/\varepsilon$ bases such that 
$\E[F(x)] \geq (1 - 1/e - \varepsilon) \cdot \max \{ f(T) \mid T \in \I \}$ holds, where $F \colon [0, 1]^V \rightarrow \mathbb{R}_+$ is the multilinear extension of $f$,
\item it uses $\displaystyle O\left( r \lambda + \frac{rn}{\lambda \varepsilon^5} \log^2 \left(\frac{n}{\varepsilon} \right) \right)$ value oracle queries, and
\item it uses $\displaystyle O\left( \frac{\lambda n}{\varepsilon^2} \log \left( \frac{n}{\varepsilon} \right) \right)$ independence oracle queries.
\end{itemize}
\end{theorem}

Suppose that $x \in B(\M)$ is a point as in Theorem~\ref{submodular_max_lemma_independence}. 
In the submodular maximization algorithm of Buchbinder-Feldman-Schwartz, 
they round $x$ to an integral solution 
with the aid of 
the swap rounding algorithm of Chekuri-Vondr\'{a}k-Zenklusen~\cite{chekuri2010dependent}  
using $O(r^2 / \varepsilon)$ independence oracle queries.
Therefore, their entire algorithm requires $\TO(r \lambda + \frac{rn}{\lambda})$ value oracle queries 
and $\TO(\lambda n + r^2)$ independence oracle queries.

\begin{remark}
 Theorem \ref{submodular_max_lemma_independence} is not explicitly stated in the paper by Buchbinder-Feldman-Schwartz~\cite{buchbinder2017comparing}, 
 because 
 they do not separately evaluate the query complexity for solving the continuous relaxation problem and 
 for the rounding algorithm. 
 Indeed, they just state that 
 the entire algorithm requires $O( \frac{r^2}{\varepsilon} + \frac{\lambda n}{\varepsilon^2} \log \left( \frac{n}{\varepsilon} \right) )$ independence oracle queries; 
 see \cite[Theorem 1.1]{buchbinder2017comparing}.
 The $O(\frac{r^2}{\varepsilon})$ term in this query complexity comes from \cite[Corollary 3.1]{buchbinder2017comparing}, 
 which states that the continuous greedy algorithm together with the rounding algorithm requires 
 $O(\frac{n}{\varepsilon^{2}} \log ( \frac{n}{\varepsilon}) + \frac{r^2}{\varepsilon})$ independence oracle queries.  
 This corollary is a direct consequence of \cite[Claim 4.4]{badanidiyuru2014fast}, whose proof shows that  
 the $O(\frac{r^2}{\varepsilon})$ term comes from the rounding algorithm, while
 the $O(\frac{n}{\varepsilon^{2}} \log ( \frac{n}{\varepsilon}))$ term comes from the continuous greedy algorithm. 
 Therefore, the $O(\frac{r^2}{\varepsilon})$ term is not needed to solve the continuous relaxation problem. 
\end{remark}




We now show that our submodular maximization algorithm with subquadratic query complexity 
is derived from Theorems~\ref{thm:main_rounding} and~\ref{submodular_max_lemma_independence}.

\begin{proof}[Proof of Theorem \ref{main_submodular_max}]
Let $\lambda = \Theta(\sqrt{r})$ and let $\varepsilon' = \varepsilon /2$. 
Note that we can compute $r$ using $O(n)$ independence oracle queries by a greedy algorithm. 
We first run the algorithm in Theorem \ref{submodular_max_lemma_independence}
with parameters $\lambda$ and $\varepsilon'$
to obtain a point $x \in B(\M)$, in which 
we use $O(\sqrt{r} n \text{ \textup{poly}}(1/\varepsilon, \log n))$ value and independence oracle queries. 
For the obtained point $x$, 
we apply our fast rounding algorithm in Theorem \ref{thm:main_rounding} with an error parameter $\varepsilon'$. 
Then, we obtain a basis $S$ of $\M$ such that 
\begin{align*}
\E[f(S)] &\geq (1 - \varepsilon') \cdot \E[F(x)] \\ 
         &\geq (1 - \varepsilon') \cdot (1 - 1/e - \varepsilon') \cdot \max \{ f(T) \mid T \in \I \} \\
         &\geq (1 - 1/e - \varepsilon) \cdot \max \{ f(T) \mid T \in \I \}. 
\end{align*}
Since $x$ is represented as a convex combination of $1/\varepsilon'$ bases of $\M$ by Theorem~\ref{submodular_max_lemma_independence}, 
our rounding algorithm requires $O(r^{3/2} \text{ \textup{poly}}(1/\varepsilon, \log n))$ independence oracle queries by Theorem~\ref{thm:main_rounding}. 
Therefore, 
we obtain a $(1 - 1/e - \varepsilon)$-approximation algorithm that uses $O(\sqrt{r} n \text{ \textup{poly}}(1/\varepsilon, \log n))$ value and independence oracle queries, which completes the proof.
\end{proof}


%% file: algorithm.tex



\section{Swap Rounding Algorithm in Previous Work}
\label{sec:preimplement}

In this section, we describe the swap rounding algorithm of Chekuri-Vondr\'{a}k-Zenklusen~\cite{chekuri2010dependent} for a matroid base polytope, which we denote \texttt{SwapRound}. 
As described in Section~\ref{subsec:overview}, 
our new rounding algorithm is based on \texttt{SwapRound}. 
In \texttt{SwapRound},  
we are given a point $x \in B(\M)$ that is represented as a convex combination of 
the characteristic vectors of $t$ bases of $\M$. 
The output is a single basis $S$ of $\M$ such that
$\E[f(S)] \geq F(x)$ for any submodular function $f\colon 2^V \to \mathbb{R}$ and its multilinear extension $F$. 
In each phase of \texttt{SwapRound}, 
we pick up two bases in the representation of $x$ and merge them into a basis. 
By applying this procedure $t-1$ times, \texttt{SwapRound} finally returns a single basis of $\M$; see Algorithm~\ref{alg:swapround}.

The procedure for merging two bases is denoted by \texttt{MergeBases} (Algorithm \ref{alg:mergebases}). 
The input of \texttt{MergeBases} consists of two bases $B_1$ and $B_2$ together with 
their coefficients $\beta_1$ and $\beta_2$. 
In the procedure, until $B_1$ and $B_2$ coincide, 
we repeatedly update $B_1$ and $B_2$ so that $|B_1 \setminus B_2|$ decreases monotonically. 
In each update of $B_1$ and $B_2$, we need 
a {\em strongly exchangeable pair} of elements, that is, 
a pair of elements $u \in B_1 \setminus B_2$
and $v \in B_2 \setminus B_1$ such that $B_1 + v - u \in \I$ and $B_2 + u - v \in \I$. 
As described in \texttt{UpdateViaStrongBasisExchange} (Algorithm~\ref{alg:update_via_strong_basis_exchange}), 
for a strongly exchangeable pair $u$ and $v$, 
we apply $B_1 \gets B_1 + v - u$ with probability $\frac{\beta_2}{\beta_1+\beta_2}$
and apply $B_2 \gets B_2 + u - v$ with the remaining probability. 
Note that, in \texttt{UpdateViaStrongBasisExchange}, 
$B_1$ and $B_2$ are updated to $B_1'$ and $B_2'$ 
so that $\E[\beta_1 \mathbf{1}_{B_1'} + \beta_2 \mathbf{1}_{B_2'}] = \beta_1 \mathbf{1}_{B_1} + \beta_2 \mathbf{1}_{B_2}$, 
which is a key property to show the validity of the algorithm.

The most time consuming part in \texttt{MergeBases} is to find 
a strongly exchangeable pair. 
By the strong basis exchange property of matroids, 
we can find such a pair of elements using $O(r)$ independence oracle queries in the following way: 
fix an element $u\in B_1 \setminus B_2$ arbitrarily and 
check the conditions for each element $v \in B_2 \setminus B_1$ one by one. 
Since we update the bases $|B_1 \setminus B_2| = O(r)$ times, 
\texttt{MergeBases} requires $O(r^2)$ independence oracle queries in total.
Hence, \texttt{SwapRound} requires $O(r^2 t)$ independence oracle queries. 

It is not clear whether we can develop an algorithm that finds 
a strongly exchangeable pair
using $o(r)$ independence oracle queries.
Therefore, their implementation of \texttt{SwapRound} is now stuck at $\Omega(r^2 t)$ independence oracle queries.

\begin{algorithm}[t] 
    $C_1 \gets B_1$ \\
    $\gamma_1 \gets \beta_1$ \\
    \For{$i = 1$ \textup{to} $t - 1$} {
        $C_{i + 1} \gets \text{\texttt{MergeBases}}(\gamma_i, C_i, \beta_{i + 1}, B_{i + 1})$ \\
        $\gamma_{i + 1} \gets \gamma_i + \beta_{i + 1}$ \\
    }
    \Return $C_t$
    \caption{\texttt{SwapRound}$(x = \sum_{i = 1}^{t} \beta_i \mathbf{1}_{B_i})$}\label{alg:swapround}
\end{algorithm}

\begin{algorithm}[t] 
    \While {$B_1 \neq B_2$} {
        Pick arbitrary $u \in B_1 \setminus B_2$ \\
        Find $v \in B_2 \setminus B_1$ such that $B_1 + v - u \in \I$ and $B_2 + u - v \in \I$ \label{line:mergebases_1} \\
        \texttt{UpdateViaStrongBasisExchange}$(\beta_1, B_1, \beta_2, B_2, v, u)$ \\
    } 
    \Return $B_1$
    \caption{\texttt{MergeBases}$(\beta_1, B_1, \beta_2, B_2)$}\label{alg:mergebases}
\end{algorithm}

\begin{algorithm}[t] 
    \KwIn{$\beta_1, \beta_2 \in \mathbb{R}_{+}$, two bases $B_1, B_2$, and elements $v \in B_2 \setminus B_1$ and $u \in B_1 \setminus B_2$  such that $B_1 + v - u \in \I$ and $B_2 + u - v \in \I$}
    Flip a coin with \texttt{Heads} probability $\displaystyle \frac{\beta_2}{\beta_1 + \beta_2}$ \\
    \If{coin flipped \textup{\texttt{Heads}}} {
        $B_1 \gets B_1 + v - u$
    } \Else {
        $B_2 \gets B_2 + u - v$
    }
    \caption{\texttt{UpdateViaStrongBasisExchange}$(\beta_1, B_1, \beta_2, B_2, v, u)$}\label{alg:update_via_strong_basis_exchange}
\end{algorithm}

\section{Faster Rounding Algorithm} \label{sec:faster_rounding_algorithm}



In this section, we present our fast rounding algorithm. 
We first show the following theorem, and then prove Theorem \ref{thm:main_rounding} using this theorem.

\begin{theorem} \label{thm:fast_swap_rounding}
There is a randomized algorithm satisfying the following conditions: 
\begin{itemize}
    \item the input consists of a matroid $\M = (V, \I)$ given as an independence oracle and 
    a point $x \in B(\M)$ represented as a convex combination of $t$ bases, 
    \item the output is a basis $S$ of $\M$ such that $\E[f(S)] \geq F(x)$ for any submodular function $f\colon 2^V \to \mathbb{R}$ and its multilinear extension $F$, and 
    \item it uses $O(r^{3/2}t \log^{3/2}(r t))$ independence oracle queries with probability at least $1 - (rt)^{-1}$. 
\end{itemize}
\end{theorem}


To show this theorem, we propose an algorithm that merges the bases one by one 
in the same way as \texttt{SwapRound}. 
Our contribution is to improve \texttt{MergeBases}, that is, 
we give a faster algorithm for merging two bases into a single basis. 
The following auxiliary graph plays an important role in our algorithm.


\begin{definition} \label{def:auxiliary_directed_graph}
Let $\M = (V, \I)$ be a matroid, and let $B_1, B_2$ be two bases of $\M$.
Then we define the bipartite directed graph $D_{\M}(B_1, B_2)$ 
whose vertex set and edge set are 
$B_1 \triangle B_2$ and 
$E_1(B_1, B_2) \cup E_2(B_1, B_2)$, respectively, where 
\begin{align*}
E_1(B_1, B_2) &= \{(u, v) \mid u \in B_1 \setminus B_2, v \in B_2 \setminus B_1, B_1 + v - u \in \I \} , \\
E_2(B_1, B_2) &=  \{(v, u) \mid u \in B_1 \setminus B_2, v \in B_2 \setminus B_1, B_2 + u - v \in \I \}.
\end{align*}
\end{definition}

In terms of this auxiliary graph, each step of \texttt{MergeBases} can be interpreted as follows: 
it finds a directed cycle of length two (or a bidirected edge) in $D_{\M}(B_1, B_2)$
and updates the bases $B_1$ and $B_2$ using this directed cycle as in \texttt{UpdateViaStrongBasisExchange}. 
Note that we use $O(r)$ independence oracle queries to find a directed cycle of length two.

A key idea in our algorithm is to focus on a directed cycle of arbitrary length in $D_{\M}(B_1, B_2)$
instead of a directed cycle of length two. 
More precisely, our contribution consists of the following two technical results. 
\begin{enumerate}
    \item We can find a directed cycle in $D_{\M}(B_1, B_2)$ using $o(r)$ independence oracle queries with high probability. 
    \item We can appropriately update the bases using a directed cycle of arbitrary length in $D_{\M}(B_1, B_2)$.  
\end{enumerate}
We discuss the first and second technical results in Sections~\ref{subsec:finddicycle} and~\ref{subsec:modificationdicycle}, respectively. 
Then, we describe the entire algorithm and give proofs for Theorems~\ref{thm:main_rounding} and~\ref{thm:fast_swap_rounding} in Section~\ref{subsec:wholealgorithm}.

\subsection{Finding a Directed Cycle}
\label{subsec:finddicycle}

The objective of this subsection is to show the following proposition, 
which states that 
we can find a directed cycle in $D_{\M}(B_1, B_2)$ using $o(r)$ independence oracle queries with high probability. 

\begin{proposition}
\label{propf:findcycle}
Suppose we are given two bases $B_1$ and $B_2$ of a matroid $\M$ and an integer $t \ge 2$. 
Then, we can find a directed cycle in $D_{\M}(B_1, B_2)$ using $O(\sqrt{r} \log^{3/2} (rt))$ independence oracle queries
with probability at least $1-(rt)^{-2}$. 
\end{proposition}

To show this proposition, we first show that a directed cycle of length two can be found efficiently 
if we have an element whose indegree is small in $D_{\M}(B_1, B_2)$. 

\begin{lemma}
\label{lem:smalldegbruteforce}
Suppose we are given two bases $B_1$ and $B_2$ of a matroid $\M$, and 
an element $a \in B_1 \triangle B_2$ 
whose indegree is $d$ in $D_{\M}(B_1, B_2)$.
Then, we can find a directed cycle of length two in $D_{\M}(B_1, B_2)$ 
using $O(d \log r)$ independence oracle queries. 
\end{lemma}

\begin{proof}
By symmetry, it suffices to consider the case when $a \in B_1 \setminus B_2$. 

We give an algorithm that finds an element $v \in B_2 \setminus B_1$
such that $B_1 +v -a \in \I$ and $B_2 + a -v \in \I$. 
In our algorithm, let $A \subseteq B_2 \setminus B_1$ denote the set of elements $v$
such that we have already checked that $B_1 +v -a \not\in \I$. 
We initialize $A = \emptyset$. 

In each step of our algorithm, by applying Lemma~\ref{lem:binary_search_ind} in which 
$u=a$, $S=B_2$, and $T = B_2 \setminus (B_1 \cup A)$, 
we can find an element $v \in T$ such that $B_2 + a -v \in \I$ if it exists. 
For such $v$, we check whether $B_1 +v -a \in \I$ holds or not. 
If $B_1 +v -a \in \I$ holds, then $a$ and $v$ induce a desired directed cycle. 
Otherwise, we add $v$ to $A$, and repeat the procedure. 

This algorithm finds a directed cycle correctly by the strong basis exchange property. 
Since we apply Lemma~\ref{lem:binary_search_ind} at most $d$ times and $|T| \le r$, 
this algorithm uses $O(d \log r)$ independence oracle queries. 
\end{proof}

We now describe our algorithm for finding a directed cycle in $D_{\M} (B_1, B_2)$. 
In our algorithm, 
we first sample $2 \sqrt{r \log (rt)}$ elements from $B_1 \setminus B_2$ (resp.~$B_2 \setminus B_1$) uniformly at random with replacement, 
where the base of the logarithm is $e$, 
and let $L$ (resp.~$R$) be the sampled vertex set, ignoring the multiplicity. 
Note that $1 \le |L| \le 2 \sqrt{r \log (rt)}$ and $1 \le |R| \le 2 \sqrt{r \log (rt)}$
as we ignore the multiplicity. 
%
Let $D'$ be the subgraph of $D_{\M} (B_1, B_2)$ induced by $L \cup R$. 

For each vertex $u$ in $D'$, 
we find a directed edge in $D'$ that enters $u$ or conclude that such a directed edge does not exist.  
This can be done by calling \texttt{FindExchangeElement} exactly once for each $u$. 
If every vertex in $D'$ has an incoming edge, then we can easily find a directed cycle in $D'$ by traversing such directed edges in the opposite direction. 
Otherwise, we pick up a vertex $a$ in $L \cup R$ that has no incoming edge in $D'$, 
and then apply Lemma~\ref{lem:smalldegbruteforce} with this vertex $a$ to find a directed cycle of length two. 

Since the correctness of this algorithm is clear, 
it suffices to analyze the independence query complexity. 
We use the following lemma in our analysis. 








\begin{lemma} \label{lem:low_degree}
Let $u \in B_1 \triangle B_2$ be an element
whose indegree in $D_{\M}(B_1, B_2)$ is at least $2 \sqrt{r \log (rt)}$. 
Then, the probability that 
$D_{\M}(B_1, B_2)$ has no directed edge from $L\cup R$ to $u$ 
 is at most $(rt)^{-4}$. 
\end{lemma}

\begin{proof}
By symmetry, it suffices to consider the case when $u \in B_1 \setminus B_2$. 
Let $N = \{v \in B_2 \setminus B_1 \mid (v, u) \in E(B_1, B_2) \}$.
%
Since $R$ is obtained by sampling $2 \sqrt{r \log (rt)}$ vertices from $B_2 \setminus B_1$ and 
$r \ge |B_2 \setminus B_1| \ge |N| \ge 2 \sqrt{r \log (rt)}$, 
we have the following:
\begin{align*}
\text{Pr}\left[ \left\{v \in R \mid (v, u) \in E(B_1, B_2)\right\} = \emptyset \right] 
& = \text{Pr}\left[ N \cap R = \emptyset \right] \\
& =  \left(1 - \frac{|N|}{|B_2 \setminus B_1|} \right)^{2 \sqrt{r \log (rt)}} \\ 
& \leq \left(1 - \frac{2 \sqrt{r \log (rt)}}{r} \right)^{2 \sqrt{r \log (rt)}} \\ 
& \leq \left(e^{-1}\right)^{4 \log (rt)} \\ 
& = (rt)^{-4}, 
\end{align*}
which completes the proof. 
\end{proof}


We are now ready to prove Proposition~\ref{propf:findcycle}. 

\begin{proof}[Proof of Proposition~\ref{propf:findcycle}]
We analyze the independence query complexity of the algorithm described above. 
First, since we call \texttt{FindExchangeElement} for each vertex $u \in L \cup R$ exactly once to find an incoming edge in $D'$, 
the number of calls of \texttt{FindExchangeElement} is $|L \cup R| = O(\sqrt{r \log (rt)})$.
Hence, by Lemma \ref{lem:binary_search_ind}, 
the number of independence oracle queries used in this part is $O(\sqrt{r \log (rt)} \log r)$.

We next analyze the number of independence oracle queries 
when there exists a vertex $a \in L \cup R$ that has no incoming edge in $D'$. 


We call a vertex $u \in L  \cup R$ \emph{bad} if 
$D'$ has no directed edge entering $u$ and $D_{\M}(B_1, B_2)$ has at least $2 \sqrt{r \log (rt)}$ directed edges entering $u$.  
By Lemma \ref{lem:low_degree}, for each $u \in L \cup R$, the vertex $u$ is bad with probability at most $(rt)^{-4}$.
Thus, by taking the union bound over all vertices in $L \cup R$, 
we see that there exists a bad vertex in $L \cup R$ with probability at most $(rt)^{-2}$.



We now consider the case where there is no bad vertex in $L \cup R$. 
Suppose that there exists a vertex $a \in L \cup R$ that has no incoming edge in $D'$. 
Then, since $a$ is not bad, the indegree of $a$ is at most $2 \sqrt{r \log (rt)}$ in $D_{\M}(B_1, B_2)$. 
Therefore, we can apply Lemma~\ref{lem:smalldegbruteforce} with $a$ 
using $O(\sqrt{r \log (rt)} \log r)$ independence oracle queries. 

Therefore, the total number of independence oracle queries used in the algorithm 
is $O(\sqrt{r \log (rt)} \log r)$ with probability at least $1 - (rt)^{-2}$,
which completes the proof.
\end{proof}

\subsection{Update with a Directed Cycle}
\label{subsec:modificationdicycle}

In this subsection, we describe how to update the bases using a directed cycle in $D_{\M}(B_1, B_2)$.  
Let $C$ be a directed cycle in $D_{\M}(B_1, B_2)$ that traverses 
$u_0, v_0, u_1, v_1, \dots ,v_{l-1}$ in this order, where $u_i \in B_1 \setminus B_2$ and $v_i \in B_2 \setminus B_1$ for each $i$. 
In our algorithm, we first choose $B_1$
with probability 
$\frac{\beta_2}{\beta_1 + \beta_2}$ and 
choose $B_2$ with the remaining probability. 
If we choose $B_1$, then
we pick up an index $i$ uniformly at random from $\{0, \ldots, l - 1\}$ and update $B_1$ by $B_1 \gets B_1 + v_i - u_i$. 
If we choose $B_2$, then
we pick up an index $i$ uniformly at random from $\{0, \ldots, l - 1\}$ and update $B_2$ by $B_2 \gets B_2 + u_{i+1} - v_{i}$, 
where we denote $u_l = u_0$. 
The pseudocode of this algorithm is shown in \texttt{UpdateWithCycle} (Algorithm~\ref{alg:update_with_cycle}). 
We note that, if the length of the directed cycle is two, then 
\texttt{UpdateWithCycle} coincides with \texttt{UpdateViaStrongBasisExchange}.


\begin{algorithm}
    \KwIn{$\beta_1, \beta_2 \in \mathbb{R}_{+}$, two bases $B_1, B_2$, and a directed cycle $C$ in the bipartite directed graph $D_{\M}(B_1, B_2)$}
    Denote by $V(C) = \{u_0, v_0, u_1, v_1, \ldots, v_{l - 1}\}$ the vertices in $C$ in this order (with $u_i \in B_1 \setminus B_2$ and $v_i \in B_2 \setminus B_1$ for each $i$) \\ 
    Flip a coin with \texttt{Heads} probability $\displaystyle \frac{\beta_2}{\beta_1 + \beta_2}$ \\
    \If{coin flipped \textup{\texttt{Heads}}} {
        Pick an index $i$ uniformly at random from $\{0, \ldots, l - 1\}$ \\
        $B_1 \gets B_1 + v_i - u_i$ 
    } \Else {
        Pick an index $i$ uniformly at random from $\{0, \ldots, l - 1\}$ \\
        $B_2 \gets B_2 + u_{i+1} - v_i$ \tcp{We define $u_l = u_0$.} 
    }
    \caption{\texttt{UpdateWithCycle}$(\beta_1, B_1, \beta_2, B_2, C)$}\label{alg:update_with_cycle}
\end{algorithm}

In order to show the validity of the algorithm, 
we use the following two lemmas. 

\begin{lemma} \label{lem:cycle_useful}
Given two bases $B_1$ and $B_2$ and a directed cycle $C$ in the bipartite directed graph $D_{\M}(B_1, B_2)$,
the procedure \texttt{\textup{UpdateWithCycle}} updates $B_1$ and $B_2$ to $B_1'$ and $B_2'$, respectively, so that $\E[\beta_1 \mathbf{1}_{B_1'} + \beta_2 \mathbf{1}_{B_2'}] = \beta_1 \mathbf{1}_{B_1} + \beta_2 \mathbf{1}_{B_2}$. 
\end{lemma}

\begin{proof}
Recall that $C$ traverses 
$u_0, v_0, u_1, v_1, \dots ,v_{l-1}$ in this order, where $u_i \in B_1 \setminus B_2$ and $v_i \in B_2 \setminus B_1$ for each $i$. 
In the procedure \texttt{UpdateWithCycle}, 
we obtain 
$B_1' = B_1 + v_i - u _i$ for some $i \in \{0, \ldots, l - 1\}$ and $B_2' = B_2$ with probability $\beta_2 / (\beta_1 + \beta_2)$, 
and 
we obtain $B_1' = B_1$ and $B'_2 = B_2 + u_{i+1} - v_i$ for some $i \in \{0, \ldots, l - 1\}$ with probability $\beta_1 / (\beta_1 + \beta_2)$. 
Thus, we have the following equation:
\begin{align*}
\E[\beta_1 \mathbf{1}_{B_1'} + \beta_2 \mathbf{1}_{B_2'}] =& \frac{\beta_2}{\beta_1 + \beta_2} \left( \beta_1 \left( \mathbf{1}_{B_1} + \frac{1}{l} \sum_{i = 0}^{l - 1} \left( \mathbf{1}_{v_i} - \mathbf{1}_{u_i} \right) \right)  + \beta_2  \mathbf{1}_{B_2} \right) \\
& + \frac{\beta_1}{\beta_1 + \beta_2} \left(  \beta_1 \mathbf{1}_{B_1} + \beta_2 \left( \mathbf{1}_{B_2} + \frac{1}{l} \sum_{i = 0}^{l - 1} \left( \mathbf{1}_{u_{i+1}} - \mathbf{1}_{v_i}  \right) \right) \right) \\
=& \beta_1 \mathbf{1}_{B_1} + \beta_2 \mathbf{1}_{B_2}. 
\end{align*}
This completes the proof. 
\end{proof}



\begin{lemma}[{\cite[Lemma VI.2]{chekuri2010dependent}}] \label{lem:helpful_from_chekuri_et_al}
Let $x \in \mathbb{R}^n_+$ be a non-negative vector and $\mathbf{X} = (X_1, \dots , X_n)$ be a non-negative vector-valued random variable satisfying the following properties:
\begin{itemize}
 \item $\E[\mathbf{X}] = x$, and 
 \item $\mathbf{X} - x$ has at most one positive coordinate and at most one negative coordinate. 
\end{itemize}
Then, we have $\E[F(\mathbf{X})] \geq F(x)$ for any function $F$ that is 
a multilinear extension of some submodular function. 
\end{lemma}

By combining these lemmas, we obtain the following proposition, 
which shows the validity of \texttt{\textup{UpdateWithCycle}}. 

\begin{proposition}
\label{prop:correctness}
Let $x = \sum_{i = 1}^{t} \beta_i \mathbf{1}_{B_i}$ be a point represented by 
a convex combination of the characteristic vectors of $t$ bases of a matroid $\M$. 
Suppose that 
the procedure \texttt{\textup{UpdateWithCycle}} updates $B_1$ and $B_2$ to $B_1'$ and $B_2'$ 
using a directed cycle in $D_{\M}(B_1, B_2)$. 
Let $B_i' = B_i$ for $i \in \{3, \dots , t\}$ and 
let $x' = \sum_{i = 1}^{t} \beta_i \mathbf{1}_{B'_i}$. 
Then, we obtain $\E[F(x')] \geq F(x)$ for any function $F$ that is 
a multilinear extension of some submodular function. 
\end{proposition}

\begin{proof}
It is obvious that $x'-x$ has at most one positive coordinate and at most one negative coordinate, 
since only two coordinate are involved in \texttt{\textup{UpdateWithCycle}}, and exactly one of them increases and the other decreases.
We also see that $\E[x'] = x$ holds by Lemma~\ref{lem:cycle_useful}. 
Therefore, Lemma~\ref{lem:helpful_from_chekuri_et_al} shows that 
$\E[F(x')] \geq F(x)$ for any function $F$ that is 
a multilinear extension of some submodular function. 
\end{proof}

\subsection{Whole Algorithm}
\label{subsec:wholealgorithm}

We now prove Theorem~\ref{thm:fast_swap_rounding} by 
giving our fast swap rounding algorithm. 
See \texttt{FastMergeBases} (Algorithm~\ref{alg:fastmergebases}) for the pseudocode of our algorithm.

\begin{proof}[Proof of Theorem \ref{thm:fast_swap_rounding}]
Suppose that $x = \sum_{i = 1}^{t} \beta_i \mathbf{1}_{B_i}$ is a point represented by 
a convex combination of the characteristic vectors of $t$ bases of a matroid $\M$. 
We pick up two bases, say $B_1$ and $B_2$, in the representation and merge them into a single basis in the following way: 
until $B_1$ and $B_2$ coincide, 
we find a directed cycle $C$ in $D_{\M}(B_1, B_2)$ using Proposition~\ref{propf:findcycle}, and update $B_1$ and $B_2$ 
by \texttt{\textup{UpdateWithCycle}} using $C$.  
Our algorithm repeats this process $t-1$ times so that 
all the bases are merged into a single basis. 

Since the correctness of this algorithm is shown by Proposition~\ref{prop:correctness}, 
it remains to analyze the independence query complexity of this rounding algorithm.

For merging two bases into a single basis, 
since we apply Proposition~\ref{propf:findcycle} at most $r$ times, we require  
$O(r^{3/2} \log^{3/2}(rt))$ independence oracle queries with probability at least $1 - r^{-1} t^{-2}$.
Furthermore, since we apply this procedure $t - 1$ times in our swap rounding algorithm, 
the entire algorithm requires 
$O(t r^{3/2} \log^{3/2}(rt))$ independence oracle queries with probability at least $1 - (rt)^{-1}$.
This completes the proof of Theorem~\ref{thm:fast_swap_rounding}. 
%
%
%
\end{proof}




\begin{algorithm}
    \While {$B_1 \neq B_2$} {
            Sample a set $L$ of $2 \sqrt{r \log (rt)}$ elements drawn uniformly and independently from $B_1 \setminus B_2$ with replacement. \\
            Sample a set $R$ of $2 \sqrt{r \log (rt)}$ elements drawn uniformly and independently from $B_2 \setminus B_1$ with replacement. \\
            $a \gets \emptyset$ \\
            $E \gets \emptyset$ \\
            \For{$u \in L$} {
                $v \gets \texttt{FindExchangeElement}(\M, B_2, u, R)$ \label{line:fastmergebases_query1} \\
                \If{$v = \emptyset$} {
                    $a \gets u$
                }
                \Else {
                    $E \gets E \cup \{(u, v)\}$
                }
            }
            \For{$v \in R$} { \label{line:fast_mergebases1}
                $u \gets \texttt{FindExchangeElement}(\M, B_1, v, L)$ \label{line:fastmergebases_query2} \\
                \If{$u = \emptyset$} {
                    $a \gets v$
                }
                \Else {
                    $E \gets E \cup \{(v, u)\}$
                }
            }

            \If{$a = \emptyset$} {
                Find a directed cycle $C$ in the bipartite directed graph $(L \cup R, E)$ \\
                \texttt{UpdateWithCycle}$(\beta_1, B_1, \beta_2, B_2, C)$ \\
            } \Else {

                \If{$a \in B_1 \setminus B_2$} {
                    $A \gets \emptyset$ \\
                    \While{$v = $ \textup{\texttt{FindExchangeElement}}$(\M, B_2, a, B_2 \setminus (B_1 \cup A))$ satisfies $v \neq \emptyset$} {
                        \label{line:fastmergebases_query3}
                        \If{$B_1 + v - a \in \I$} { \label{line:fastmergebases_query5}
                            \texttt{UpdateViaStrongBasisExchange}$(\beta_1, B_1, \beta_2, B_2, v, a)$ \\
                            \Break
                        }
                        $A \gets A + v$ \\
                    }
                } \Else {
                    $A \gets \emptyset$ \\
                    \While{$u = $ \textup{\texttt{FindExchangeElement}}$(\M, B_1, a, B_1 \setminus (B_2 \cup A))$ satisfies $u \neq \emptyset$} {
                        \label{line:fastmergebases_query4}
                        \If{$B_2 + u - a \in \I$} { \label{line:fastmergebases_query6}
                            \texttt{UpdateViaStrongBasisExchange}$(\beta_1, B_1, \beta_2, B_2, a, u)$ \\
                            \Break
                        }
                        $A \gets A + u$ \\
                    }
                }
            }

    } 
    \Return $B_1$
    \caption{\texttt{FastMergeBases}$(\beta_1, B_1, \beta_2, B_2)$}\label{alg:fastmergebases}
\end{algorithm}

We can remove the condition `with probability at least $1 - (rt)^{-1}$' 
by losing a sufficiently small approximation factor $\varepsilon >0$. 
That is, we obtain Theorem~\ref{thm:main_rounding}, which we restate here. 

\rounding*


\begin{proof}
Recall that the algorithm in Theorem~\ref{thm:fast_swap_rounding} (Algorithm~\ref{alg:fastmergebases})
uses $O(r^{3/2}t \log^{3/2}(r t))$ independence oracle queries with probability at least $1 - (rt)^{-1}$.
If Algorithm~\ref{alg:fastmergebases} returns a basis using $O(r^{3/2}t \log^{3/2}(r t))$ independence oracle queries, then 
we say that it {\em succeeds}. Otherwise, we say that it {\em fails}. 
By a slight modification, when the algorithm fails, 
we suppose that it uses $O(r^{3/2}t \log^{3/2}(r t))$ independence oracle queries
and terminates without returning a basis. 
This modified algorithm is denoted by Algorithm~\ref{alg:fastmergebases}'. 
Note that Algorithm~\ref{alg:fastmergebases}' fails with probability at most $(rt)^{-1}$. 

Let $q := \lceil \log_{(rt)^{-1}} \varepsilon \rceil = \lceil \frac{\log (1/\varepsilon)}{\log rt} \rceil = O \left( \frac{\log (rt/\varepsilon)}{\log rt} \right)$. 
In our algorithm, 
we run Algorithm~\ref{alg:fastmergebases}' $q$ times. 
If at least one execution of Algorithm~\ref{alg:fastmergebases}' succeeds, then 
our algorithm returns a basis that is obtained in the first successful execution of Algorithm~\ref{alg:fastmergebases}'. 
If all the executions of Algorithm~\ref{alg:fastmergebases}' fail, then our algorithm returns an arbitrary basis. 
Then, we use $O(r^{3/2}t \log^{3/2}(\frac{r t}{\varepsilon}))$ independence oracle queries in total. 
Furthermore, the probability that 
all the executions of Algorithm~\ref{alg:fastmergebases}' fail 
is at most $(rt)^{-q} \le \varepsilon$. 
Therefore, the output $S$ satisfies 
$\E[f(S)] \geq (1 - \varepsilon) F(x)$
for any submodular function $f$ and its multilinear extension $F$. 
This completes the proof. 
\end{proof}






%% file: appendix_other_oracle.tex
\section{Submodular Maximization with Rank Oracle} \label{sec:other_oracle}

\input{rank_oracle.tex}


%% file: rank_oracle.tex

In this section, we present a fast submodular maximization algorithm in the rank oracle model 
and prove Theorem~\ref{rank_submodular_max}.
In the rank oracle setting, 
the input consists of 
a monotone submodular set function $f \colon 2^V \to \mathbb{R}_+$ given as a value oracle, and 
a matroid $\M = (V, \I)$ given as a rank oracle. 
The objective is to maximize $f(S)$ subject to $S \in \I$.  
We restate Theorem~\ref{rank_submodular_max} here. 

\ranksubmodularmax*

In the same way as the independence oracle setting, 
our algorithm is based on continuous relaxation and rounding technique. 


\paragraph*{Algorithm for the Continuous Relaxation Problem.}
Let $F$ be the multilinear extension of $f$ and let $P(\M)$ be the matroid polytope of $\M$.
Ene-Nguy$\tilde{{\hat{\text{e}}}}$n~\cite{ene2018towards} presented a framework  to solve the continuous optimization problem $\max_{x \in P(\M)} F(x)$ in near-linear time for several important classes of matroids.
In their algorithm, 
they use a data structure for maintaining a maximum weight basis of the matroid, where each element has a weight and the weights are updated.
In each update, the weight of exactly one element decreases, while all the other weights do not change.    
The data structure supports an operation that decreases the weight of an element and updates the current basis to a maximum weight basis with respect to the updated weights.
This operation is called the {\em maximum weight basis data structure operation}. 
With this terminology, their result is stated as follows. 



\begin{lemma}[{follows from Lemmas 8 and 9 in the arXiv version of \cite{ene2018towards}}] \label{submodular_max_lemma_rank}
Given a non-negative monotone submodular function $f \colon 2^{V} \rightarrow \mathbb{R}_+$, a matroid $\M = (V, \I)$ of rank $r$, and a parameter $\varepsilon > 0$,
there is a randomized algorithm satisfying the following conditions:
\begin{itemize}
\item the algorithm finds a point $x \in P(\M)$ represented as a convex combination of $1/\varepsilon$ bases 
such that $\E[F(x)] \geq (1 - 1/e - \varepsilon) \cdot \max \{ f(T) \mid T \in \I \}$, where $F \colon [0, 1]^V \rightarrow \mathbb{R}_+$ is the multilinear extension of $f$, 
\item it uses $\displaystyle O\left( \frac{n}{ \varepsilon^5} \log^2 \left(\frac{n}{\varepsilon} \right) \right)$ value oracle queries,
\item it uses $\displaystyle O \left( \frac{n}{\varepsilon} \log \left( \frac{n}{\varepsilon} \right) \right)$ independence oracle queries, and
\item it uses $\displaystyle O \left( \frac{r}{\varepsilon} \log^2 \left( \frac{n}{\varepsilon} \right) \right)$ maximum weight basis data structure operations.
\end{itemize}
\end{lemma}

\paragraph*{Maximum Weight Basis Data Structure Operation.}



To implement a maximum weight basis data structure operation by using rank oracle queries efficiently, 
we use the following lemma obtained by the binary search technique of Nguy$\tilde{{\hat{\text{e}}}}$n~\cite{nguyen2019note} and Chakrabarty et al.~\cite{chakrabarty2019faster}; 
see also {\cite[Lemma 2]{u2022subquadratic}}. 


\begin{lemma}[{\cite[Lemma 10]{chakrabarty2019faster}}; see also {\cite[Lemma 2]{u2022subquadratic}} and \cite{blikstad2021breaking_STOC}] \label{lem:binary_search_rank} 
There is an algorithm \textup{\texttt{FindFreeElement}} which, 
given a matroid $\M = (V, \mathcal{I})$, a weight function $w \colon V \rightarrow \mathbb{R}$, an independent set $S \in \mathcal{I}$, and $T \subseteq V \setminus S$, finds an element $u \in T$ maximizing $w(u)$ such that $S + u \in \mathcal{I}$ or otherwise determines that no such element exists, and uses $O(\log |T|)$ rank oracle queries.
\end{lemma}




This lemma shows that the maximum weight basis data structure operation can be easily implemented in the rank oracle model as follows. 

\begin{lemma}
\label{lem:implementdatastructure}
Let $\M = (V, \mathcal{I})$ be a matroid given as a rank oracle and
let $w \colon V \rightarrow \mathbb{R}$ be a weight function.  
Let $w' \colon V \rightarrow \mathbb{R}$ be a weight function 
such that $w'(v) < w(v)$ for some $v \in V$ and $w'(u) = w(u)$ for any $u \in V - v$. 
Given a maximum weight basis $B$ of $\M$ with respect to $w$, 
we can compute a maximum weight basis $B'$ of $\M$ with respect to $w'$
using $\tO(1)$ rank oracle queries. 
\end{lemma}

\begin{proof}
If $v \not\in B$, then $B' := B$ is a desired basis, because $w'(v) < w(v)$. 
Otherwise, we apply 
\texttt{FindFreeElement} with the weight function $w'$
in which $S = B- v$ and $T = (V \setminus B) \cup \{ v \}$. 
Let $u$ be the element found by the procedure (possibly, $u=v$). 
Then our algorithm returns a basis $B' := B - v + u$, 
which is a maximum weight basis with respect to $w'$ 
(see e.g.,~\cite[Lemma 3.1]{hommelsheim2023recoverable} and  
Section 6 of the arXiv version of \cite{blikstad2023fast}).
By Lemma \ref{lem:binary_search_rank}, this algorithm requires $\tO(1)$ rank oracle queries.
\end{proof}


\paragraph*{Putting Them Together (Proof of Theorem \ref{rank_submodular_max}).}

We now prove Theorem \ref{rank_submodular_max}. 
Lemma~\ref{lem:implementdatastructure} shows that
we can execute the maximum weight basis data structure operation
using $\tO(1)$ rank oracle queries
without a sophisticated data structure. 
Hence, by Lemma \ref{submodular_max_lemma_rank}, 
we can solve the continuous optimization problem $\max_{x \in P(\M)} F(x)$ using $\TO(n)$ value and rank oracle queries, 
where we note that the rank oracle is more powerful than the independence oracle. 

For the obtained point $x$, 
we apply our fast rounding algorithm given in Theorem \ref{thm:main_rounding} to obtain an integral solution. 
Note again that the rank oracle is more powerful than the independence oracle, 
and hence this rounding algorithm requires $\TO(r^{3/2})$ value and rank oracle queries. 

By replacing $\varepsilon$ with $\varepsilon / 2$ in the same way as in the proof of Theorem \ref{main_submodular_max}, 
we obtain a $(1 - 1/e - \varepsilon)$-approximation algorithm that uses $\TO(n + r^{3/2})$ value and rank oracle queries, which completes the proof.
%
%
%
%
%
%
\qed